\documentclass[conference,10]{IEEEtran}

\input{def.tex}
\DeclareSymbolFont{matha}{OML}{txmi}{m}{it}
\DeclareMathSymbol{\varv}{\mathord}{matha}{118}

\renewcommand{\baselinestretch}{0.93}

\begin{document}

\title{Impact of Transceiver Impairments on the Capacity of Dual-Hop Relay
Massive MIMO  Systems}
\author{Anastasios K. Papazafeiropoulos$^{*}$, Shree Krishna Sharma$^{\dag}$, and Symeon Chatzinotas$^{\dag}$\vspace{2mm} \\
$^{*}$Communications and Signal Processing Group, Imperial College London, London, U.K.\\
$^{\dag}$SnT - securityandtrust.lu, University of Luxembourg, Luxembourg\\
Email: a.papazafeiropoulos@imperial.ac.uk, \{shree.sharma, symeon.chatzinotas\}@uni.lu
 }

\maketitle

\vspace{-2cm}

\begin{abstract}
Despite the deleterious effect of hardware impairments on communication systems, most prior works have not investigated their impact on widely used relay systems. Most importantly, the application of inexpensive transceivers, being prone to hardware impairments, is the most cost-efficient way for the implementation of massive multiple-input multiple-output (MIMO) systems. Consequently, the direction of this paper is towards the investigation of the impact of hardware impairments on MIMO relay networks with large number of antennas.  Specifically, we obtain the general expression for the ergodic capacity of dual-hop (DH) amplify-and-forward (AF) relay systems. Next, given the advantages of the free probability (FP) theory with comparison to other known techniques in the area of large random matrix theory, we pursue a large limit analysis in terms of number of antennas and users by shedding light to the behavior of relay systems inflicted by hardware impairments.
\end{abstract}

\section{Introduction}
\label{sec:1}
Since the publication of the seminal papers of Telatar and Foschini~\cite{Telatar1999,Foschini1998}, showing the linear growth of the channel capacity by increasing the number of transmit and receive antennas,  multiple-input multiple-output (MIMO) systems have attracted a tremendous interest. Especially, the demand for a thousand-fold higher capacity in 5G systems has brought to the forefront a promising technique with numerous advantages, known as massive MIMO, where the base station (BS) includes a very large number of antennas~\cite{Marzetta2010}. The research  on massive MIMO has been approached mostly by applying tools from large random matrix theory such as the Silverstein's fixed-point equation and the technique of deterministic equivalents~\cite{Couillet2011,Papazafeiropoulos2014,Papazafeiropoulos2015a,Papazafeiropoulos2015}. However, most works in conventional and massive MIMO  have been based on the strong assumption of using perfect hardware in the radio-frequency (RF) chains.

Indisputably, in practical systems, numerous detrimental effects such as I/Q imbalance~\cite{Qi2010} and high-power amplifier nonlinearities~\cite{Qi2012}, appear and result to the degradation of the performance of MIMO systems. Despite the effort for mitigation of the arising impairments by application of calibration schemes at the transmitter and/or compensation algorithms at the receiver~\cite{Schenk2008}, residual distortions remain because of several reasons. For example, imperfect parameters estimation and inaccurate models  prove to be incapable to hinder the total infliction  of the system's performance. Not to mention that the cost-efficiency of the suggested for 5G massive MIMO technology rests on the application of inexpensive hardware, which will make the deleterious effect of the residual impairments more pronounced.

Disregarding the importance for study of the effects of the residual transceiver impairments, the number of relevant works is limited. For instance, experimental results modeling the residual hardware impairments at the transmitter and the study of their impact on certain MIMO detectors such as zero-forcing took place in~\cite{Studer2010}. Regarding the channel capacity,~\cite{Bjoernson2013} elaborated on the derivation of high signal-to-noise ratio (SNR) ceilings by considering only transmitter impairments, while in~\cite{Zhang2014} the authors extended the analysis to arbitrary SNR values, but most importantly, by including receiver impairments as well. Moreover, the authors in~\cite{Bjornson2013} considered this kind of impairments in dual-hop (DH) amplify-and-forward (AF) relay systems, which have attracted a lot of attention recently due to their performance benefits in terms of coverage extension, spatial diversity gains, etc.~\cite{Yang2009}. Unlikely,  they considered only the outage probability and  some simple capacity upper bounds in the simplistic case of single antenna systems. However, a thorough analysis of the capacity of relay systems in the case of multiple antennas lacks from the literature.  

This work covers the arising need for the assessment of the impact of the residual hardware impairments on the capacity of DH AF systems,  when their size becomes large in terms of number of antennas and users. In fact, to the best of our knowledge, this is the first paper which studies the effect of residual impairments in relay systems with multiple antennas. Specifically, we present a new insightful expression for the ergodic capacity of DH AF systems under the presence of residual hardware impairments for arbitrary SNR values. Compared to the existing literature, we pursue a free probability (FP) analysis~\cite{Couillet2011,DBLP:journals/twc/ChatzinotasIH09}, which requires just a polynomial solution instead of fixed-point equations, and provide a thorough characterization of the impact of the residual transceiver impairments on the capacity of DH AF systems in the large system limit.

The remainder of this paper is organized as follows: Section \ref{sec:2} presents the system model of a DH AF MIMO system with residual hardware impairments. Section \ref{sec:3} provides the theoretical analysis for the ergodic capacity of the considered system, while Section \ref{sec:4} presents asymptotic capacity expressions. Subsequently, Section \ref{sec:5} evaluates the performance of the considered system with the help of numerical results. Finally, Section \ref{sec:6} concludes the paper. Appendices include some preliminaries on random matrix theory and certain proofs. 

\section{System Model}
\label{sec:2}
Suppose an ideal DH AF relay channel with $K$ single antenna non-cooperative users, desiring to  communicate with a distant $N$-antennas BS by first contacting an intermediate relay including an array of $M$ antennas (first hop). In other words, a single-input multiple-output multiple access channel (SIMO MAC), i.e., users-relay, is followed by a point to point MIMO channel (relay-BS). The BS is assumed to be aware of the total system channel state information (CSI) and the statistics of the distortion noises, while both the users and the relay have no CSI knowledge during their transmission. The received signals by the relay and the BS are expressed as 
\begin{align}
\by_{1}&=\bH_{1}\bx_{1}+\bz_{1},\label{BasicSystemModell}\\
\by_{2}&=\bH_{2}\sqrt{\nu}\by_{1}+\bz_{2}\nn\\
&=\sqrt{\nu}\bH_{2}\bH_{1}\bx_{1}+\sqrt{\nu}\bH_{2}\bz_{1}+\bz_{2}\label{BasicSystemModel2},
\end{align}
where~\eqref{BasicSystemModell} and~\eqref{BasicSystemModel2} describe the users-relay and relay-BS input-output signal models, respectively. Specifically, $\by_{1}$ and $\by_{2}$ as well as $\bz_{1}\sim \cC\cN(\b0,\bI_{M})$ and $\bz_{2}\sim\cC\cN(\b0,\bI_{N})$ denote the received signals as well as the additive white Gaussian noise (AWGN) vectors at the relay and BS, respectively. Both channels representing the two hops assume  Rayleigh fast-fading channels, expressed by Gaussian matrices with independent and identically distributed (i.i.d.) complex circularly symmetric elements. Hence, $\bH_{1}  \in \bbC^{M \times K}\sim\cC\cN(\b0,\Id_{M}\otimes \Id_{K})$  is  the concatenated channel matrix between the $K$ users and the relay  exhibiting flat-fading, while  $\bH_{2}  \in \bbC^{N \times M}\sim\cC\cN(\b0,\Id_{N}\otimes \Id_{M})$ describes the channel matrix of the second hop.  In addition,  $\bx_{1}\in \mathcal{C}^{K\times 1}$ is the Gaussian vector  of symbols simultaneously transmitted by the $K$ users with $\EE\left[\bx_{1}\bx_{1}^{\H}\right]=\bQ_{1}=\frac{\rho}{K}\Id_{K}$, i.e., the individual signal-to-noise-ratio (SNR) equals to $\mu=\frac{\rho}{K}$. Note that before forwarding the received signal $\by_{1}$ at the relay, we have assumed that it is amplified by $\nu=\frac{\al}{M\left( 1+\rho \right)}$, where we have placed a per relay-antenna fixed power constraint $\frac{\al}{M}$, i,e., $\EE\left[\|\sqrt{\nu}\by_1\|^2\right]\le \al$, where the expectation is taken over all random variables.

Unfortunately, in practice, the transmitter, the relay, and the BS appear certain inevitable impairments such as I/Q imbalance~\cite{Schenk2008}. Although, mitigation schemes are incorporated in both the transmitter and the receiver, residual impairments  still emerge by means of additive distortion noises~\cite{Schenk2008,Studer2010}.  
Taking this into consideration, in each part of the system a transmit and/or receive impairment exists that  causes: 1) a mismatch between the intended signal and what is actually transmitted during the transmit processing, 2) and/or  a distortion of the received signal at the receiver side.

 Introduction of the residual additive transceiver impairments to~\eqref{BasicSystemModell} and~\eqref{BasicSystemModel2} provides the more general channel models for the respective links 
\begin{align}
&\!\!\!\by_{1}\!=\!\bH_{1}\!\left( \bx_{1}\!+\!\etv_{\mathrm{t}_1} \right)\!+\!\etv_{\mathrm{r}_1}\!+\!\bz_{1}\label{BasicSystemModel},\\
&\!\!\!\by_{2}\!=\!\bH_{2}\left( \sqrt{\nu}\by_{1}\!+\!\etv_{\mathrm{t}_2} \right)\!+\!\etv_{\mathrm{r}_2}+\bz_{2}\nn\\
&\!\!\!\!=\!\sqrt{\nu}\bH_{2}\bH_{1}\!\!\left( \bx_{1}\!+\!\etv_{\mathrm{t}_1}\! \right)\!+\!\bH_{2}\!\left( \! \sqrt{\nu}\!\left(\!\etv_{\mathrm{r}_1}\!\!+\!\bz_{1} \right)\!+\!\etv_{\mathrm{t}_2} \right)\!+\!\etv_{\mathrm{r}_2}\!+\!\bz_{2}\label{BasicSystemMode2},
\end{align}
where the additive terms $\etv_{\mathrm{t}i}$ and $\etv_{\mathrm{r}i}$ for $i=1,2$ are the distortion noises coming from the residual impairments in the transmitter and receiver of link $i$, respectively. Interestingly, this model allows to quantify the impact of the additive residual transceiver impairments, described in~\cite{Schenk2008,Studer2010}, on a DH AF system. Generally, the transmitter and the receiver distortion noises for the $i$th link are modeled as Gaussian distributed, where their average power is proportional to the average signal power, as shown by measurement results~\cite{Studer2010}. Mathematically speaking, we have 
\begin{align} 
 \etv_{\mathrm{t}_i}&\sim \cC\cN(\b0,\delta_{\mathrm{t}_i}^{2}\mathrm{diag}\left( q_{\mathrm{i}_1},\ldots,q_{T_\mathrm{i}} \right)),\\
 \etv_{\mathrm{r}_i}&\sim \cC\cN(\b0,\delta_{\mathrm{r}_i}^{2}\tr\left( \bQ_{i} \right)\Id_{R_\mathrm{i}})
\end{align}
with $T_\mathrm{i}$ and $R_\mathrm{i}$ being the numbers of transmit and receive antennas of link $i$, i.e., $T_\mathrm{1}=K\mathrm{,}~T_\mathrm{2}=M$ and $R_\mathrm{1}=M\mathrm{,}~R_\mathrm{2}=N$, while $\bQ_{i}$ is the transmit covariance matrix of the corresponding link with diagonal elements $q_{\mathrm{i}_1},\ldots,q_{T_\mathrm{i}}$. Moreover, $\delta_{\mathrm{t}_i}^{2}$ and $\delta_{\mathrm{t}_i}^{2}$ are proportionality parameters describing the
severity of the residual impairments in the transmitter and the receiver of link $i$. Especially, in practical applications, these parameters appear  as the error vector magnitudes (EVM) at each transceiver side~\cite{Holma2011}. Obviously, as far as the first hop is concerned, the additive transceiver impairments are expressed as
\begin{align}
 \etv_{\mathrm{t}_1}&\sim \cC\cN(\b0,\delta_{\mathrm{t}_1}^{2}\frac{\rho}{K}\Id_{K}),\\
 \etv_{\mathrm{r}_1}&\sim \cC\cN(\b0,\delta_{\mathrm{r}_1}^{2}\rho\Id_{M}).
\end{align}
Given that the input signal for the second hop is $\sqrt{\nu}\by_{1}$, the corresponding input covariance matrix is 
\begin{align}
 \bQ_{2}=\nu\EE\left[ \by_{1}\by_{1}^{\H}\right]&=\nu K \left( \mu+\delta_{\mathrm{t}_{1}}^{2} \mu+\delta_{\mathrm{r}_{1}}^{2}\mu +\frac{1}{K}\right)\Id_{M}\nn\\
 &=\tilde{\mu}\nu K\Id_{M},\label{constraint}
\end{align}
where $\tilde{\mu}=\left( \mu+\delta_{\mathrm{t}_{1}}^{2} \mu+\delta_{\mathrm{r}_{1}}^{2}\mu +\frac{1}{K}\right)$. Note that now, $\nu=\frac{\al}{K M \tilde{\mu}}$, after accounting for fixed gain relaying.  
Thus, the additive transceiver impairments for the second hop take the form
\begin{align}
  \etv_{\mathrm{t}_2}&\sim \cC\cN(\b0,\delta_{\mathrm{t}2}^{2}\tilde{\mu}\nu K \Id_{M}),\\
 \etv_{\mathrm{r}_2}&\sim \cC\cN(\b0,\delta_{\mathrm{r}_2}^{2}\tilde{\mu}\nu K M \Id_{N}).
\end{align}

\section{Ergodic Capacity Analysis}
\label{sec:3}
The capacity per receive antenna of this channel model  is given by  the following lemma, which takes~\eqref{BasicSystemMode2} into account.
\begin{lemma}\label{capacityGeneral}
The capacity per receive antenna of a DH  AF system in the presence of i.i.d. Rayleigh fading with residual additive transceiver hardware impairments under per user power constraints $[\bQ_{1}]_{k,k} \le \mu, \forall k=1\ldots K$ and~\eqref{constraint} is given by
\begin{align}
&\mathrm{C}=\frac{1}{N}\EE\left[\ln\det\left(\Id_{N}+\frac{\mu\nu}{B}\bH_2\bH_1\bH_1^H\bH_2^H\bPhi^{-1}\right)\right]\label{eq: first form}
\\&=\underbrace{\frac{1}{N}\EE\left[\ln\det\left(\bPhi+\frac{\mu\nu}{B}\bH_2\bH_1\bH_1^H{\bH}_2^H\right)\right]}_{\mathrm{C}_1}\nn\\
&-\underbrace{\frac{1}{N}\EE\left[\ln\det\left(\bPhi\right)\right]}_{\mathrm{C}_2},
\label{eq: generic capacity}
\end{align}
where $\bPhi=f_{2}\bH_{2}\bH_{1}\bH_{1}^{\H}\bH_{2}^{\H}+f_{3}\bH_{2}\bH_{2}^{\H}+\Id_{N}$ with $B=\delta_{\mathrm{r}_{2}}^{2}\tilde{\mu}\nu K M +1$, $f_1=\frac{f_{2}+f_{4}}{f_{3}}$,  $f_2={f_{4}\delta_{\mathrm{t}_{1}}^{2}}$,  $f_3=\frac{\nu\left( \delta_{\mathrm{t}_{2}}^{2} \tilde{\mu} K  +\delta_{\mathrm{r}_{1}}^{2} \mu K +1\right)}{B}$, and $f_{4}=\frac{\mu\nu}{B}$.  
\end{lemma}
\begin{proof}
Given any channel realizations $\bH_{1}, \bH_{2}$ and transmit signal covariance matrices $\bQ_{1}$ and $\bQ_{2}$ at the user and relay sides, a close observation of~\eqref{BasicSystemMode2} shows that it is  an instance of the standard DH AF system model described by~\eqref{BasicSystemModel2}, but with a different noise covariance given by 
\begin{align}
 \bPhi&=\nu\delta_{\mathrm{t}_{1}}^{2}\bH_{2}\bH_{1}\mathrm{diag}\left( q_{\mathrm{1}_1},\ldots,q_{K} \right)\bH_{1}^{\H}\bH_{2}^{\H}\nn\\
 &+\bH_{2}\left( \left( \nu\delta_{\mathrm{r}_{1}}^{2}\tr \bQ_{1}+\nu \right)\Id_{N}+\nu\delta_{\mathrm{t}_{2}}^{2} \mathrm{diag}\left( q_{\mathrm{1}_1},\ldots,q_{M} \right) \right)\bH_{2}^{\H}\nn\\
 &+\left( \delta_{\mathrm{r}_{2}}^{2}\tr \bQ_{2}+1 \right)\Id_{N}.
 \end{align}
Taking into account for the optimality of the input signal $\bx_{1}$ because it is Gaussian distributed with covariance matrix $\bQ_1=\frac{\rho}{K}\Id_{K}$, the proof is concluded.
\end{proof}

Based on Lemma~\ref{capacityGeneral}, we are able to  investigate the impact of the additive transceiver impairments on DH AF systems in the case of  infinitely  large system dimensions.

\begin{remark}
Despite the resemblance of the ergodic capacity with transceiver impairments, given by~\eqref{eq: first form}, with the conventional ergodic capacity of a DH AF system~\cite[Eq.~2]{Jin2010}, this paper shows the fundamental differences that arise because the noise covariance matrix now depends on the combination of the channel matrices $\bH_{1},\bH_{2}$. 
\end{remark}

Employing the property $\det(\Id+\bA\bB)=\det(\Id+\bB \bA)$, $\mathrm{C}_1, \mathrm{C}_2$ can be  alternatively written as
\begin{align}
\!\!\!\mathrm{C}_1&\!=\!\frac{1}{N}\EE\!\left[\ln\det\left(\Id_{M}\!+\!f_{3}\bH_2^\H\bH_2\left( \Id_{M}\! +\!f_{1}\bH_{1}\bH_{1}^{\H}\right)\right)\right]\label{eq:first form2}\\
\!\!\!\mathrm{C}_2&\!=\!\frac{1}{N}\EE\!\left[\ln\det\!\left(\!\Id_{M}\!+\!f_3\bH_2^H\bH_2  \! \left(\! \Id_M\!+\!\frac{f_{2}}{f_{3}}\bH_1\bH_1^H \!\right)\! \!\right)\!\right]\!\!. \label{eq:second form2}
\end{align}

\section{Asymptotic Performance Analysis}
\label{sec:4}
This section presents the main results regarding the system performance in the large-antenna regime. Nevertheless, we account also for the scenario, where the number of users increases infinitely. Given that our interest is focused on channel matrices with dimensions tending to infinity, we employ tools from large RMT. Among the advantages of the ensuing analysis, we mention the achievement of deterministic results that make Monte Carlo simulations unnecessary. Moreover,  the asymptotic analysis can be quite accurate even for realistic system dimensions, while its convergence is rather fast as the channel matrices grow large. Thus, after defining $\beta\triangleq \frac{K}{M}$ and $\gamma\triangleq \frac{N}{M}$, the channel  capacity of the system under study is given by the following theorem.
\begin{Theorem}\label{Capacity}
 The capacity of a DH AF MIMO system in the presence of  i.i.d. Rayleigh fading channels with additive transceiver impairments, when the number of transmit users $K$ as well as relay and BS antennas ($M$ and $N$) tend to infinity with a given ratio, is given by
 \begin{align}
 \!\!\!\!\!\mathcal{C}\!\rightarrow\!\frac{1}{\gamma}\!\int_0^\infty \!\!\ln\left(1\!+\!f_{3}M x\right)\!\!\left(\!\! f^\infty_{\mathbf{K}_{\!f_{1}}/M}\!\left(x\right)\!-\!f^\infty_{\mathbf{K}_{\!\frac{f_{2}}{f_{3}}}/M}\!\left(x\right)\! \!\right)\!\!\mathrm{d}x,\!\!\!\!
\label{eq: capacity aepdf 1}
 \end{align}
where the  asymptotic eigenvalue probability density functions (a.e.p.d.f.)   $f^\infty_{\mathbf{K}_{f_1}/M}$ and $ f^\infty_{\mathbf{K}_{f_2/f_3}/M}$ are obtained by  the imaginary part of the corresponding Stieltjes transform \(\mathcal{S}\) for real arguments.
\end{Theorem}
\begin{proof}
See Appendix~\ref{ProofOfCapacity}.
\end{proof}

 \section{Numerical Results}
 \label{sec:5}
In order to validate our theoretical analysis, Fig. \ref{fig_1} provides the a.e.p.d.f. of $\mathbf{K}_{\alpha}/M$ given by~\eqref{ka}. In particular, the histogram represents the p.d.f. of the matrix $\mathbf{K}_{\alpha}/M$
calculated numerically based on Mote Carlo (MC) simulations. Further, the solid line denotes the a.e.p.d.f. obtained by solving the polymonial (\ref{eq: polK}) for Stieltjes transform, and then applying Lemma 4.
From the result, we can observe a perfect agreement between the results obtained from theoretical analysis and MC simulations.

\begin{figure} [!h]
\begin{center}
\includegraphics[width=2.9 in]{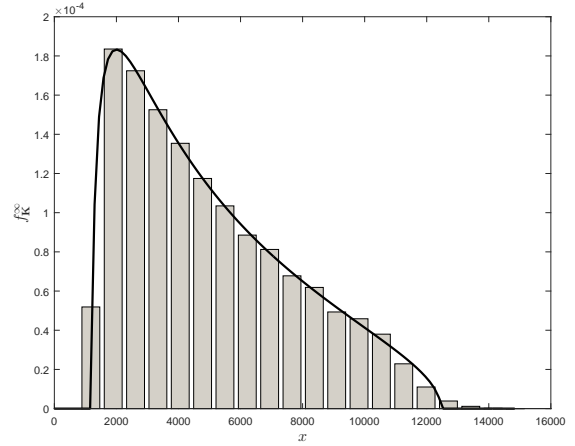}
\caption{\footnotesize{A.e.p.d.f. of $\mathbf{K}_{\alpha}/M$ ($\mu=\nu=$ 20 dB, $\beta=5$, $\gamma=10$, $N=100$, $M=10$, $K=50$)}}
\label{fig_1}
\end{center}
\vspace{-5 pt}
\end{figure}

In Fig. \ref{fig_2and3combined}, we plot the theoretical and simulated per-antenna ergodic capacities versus the   transmit SNR, i.e., $\mu$ for the following two cases: (i) without impairments, and (ii) with impairments on transmitter and receiver of both links. From the figure, it can be noted that theoretical and simulated capacity curves for both the considered cases match perfectly. Furthermore, the per-antenna capacity increases with the increase in the value of $\mu$ in the absence of impairments, i.e., $\delta_{t_1}=\delta_{t_2}=\delta_{r_1}=\delta_{r_2}=0$ as expected. Moreover, the most important observation is that the per-antenna capacity saturates after a certain value of $\mu$ in the presence of impairments. The trend of per-antenna capacity decrease with respect to $\mu$ in Fig. \ref{fig_2and3combined} is well aligned with the result obtained in \cite{Zhang2014} for the case of MIMO systems. 
However, for the considered scenario in this paper, an early saturation of the capacity in the presence of impairments is noted due to the  introduction of the relay node impairments. Nevertheless, in Fig. \ref{fig_2and3combined}, we also illustrate the effect of different values of impairments on the capacity considering the values of $\delta_{t_1}=\delta_{t_2}=\delta_{r_1}=\delta_{r_2}=\delta$ as $0.01$, $0.08$ and $0.15$. Specifically, it can be observed that with the increase in the value of impairments, the saturation point appears earlier, i.e., at  lower values of $\mu$.    


\begin{figure}[!h]
\begin{center}
\includegraphics[width=2.9 in]{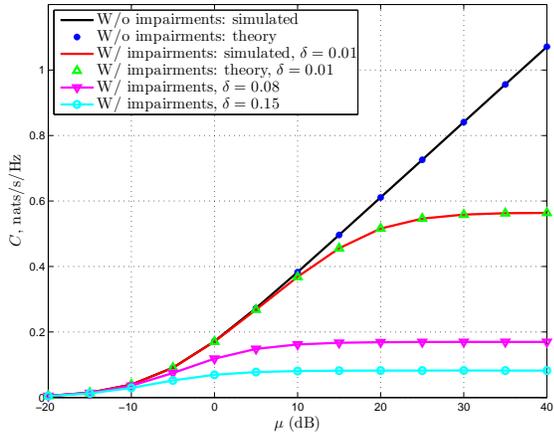}
\caption{\footnotesize{Per-antenna ergodic capacity versus $\mu$ 
($\nu=20$ dB, $\beta=5$, $\gamma=10$, $N=100$, $M=10$, $K=50$, $\delta_{t_1}=\delta_{t_2}=\delta_{r_1}=\delta_{r_2}=\delta$)}}
\label{fig_2and3combined}
\end{center}
\vspace{-5 pt}
\end{figure}

Figures \ref{fig_6and7combined1} and \ref{fig_6and7combined2} present the per-antenna capacity versus SNR levels $\mu$, $\nu$ in the presence and the absence of impairments, respectively. It can be observed that in the absence of impairments, the capacity increases monotonically with both $\mu$ and $\nu$, with the slope being steeper for the case of $\mu$. However, in the presence of impairments, clear saturation points can be noted with the increase in the values of $\mu$ and $\nu$ after  certain values. 

As far as Figs. \ref{fig_8and9combined1} and \ref{fig_8and9combined2} are concerned, we plot the per-antenna capacity versus the channel dimensions $\gamma$ and $\beta$ in the absence and the presence of channel impairments, respectively. In both cases, the capacity increases monotonically with both $\beta$ and $\frac{1}{\gamma}$, however, the slope with respect to $\frac{1}{\gamma}$ is steeper as compared to the slope with $\beta$. This trend remains almost the same in the presence of impairments, but the slope of the capacity curve with respect to $\frac{1}{\gamma}$ in Fig. \ref{fig_8and9combined2} is observed to be less steeper  than in Fig. \ref{fig_8and9combined1} at higher values of $\frac{1}{\gamma}$.

\begin{figure}
\begin{center}
\subfigure[]{
\includegraphics[scale=.39]{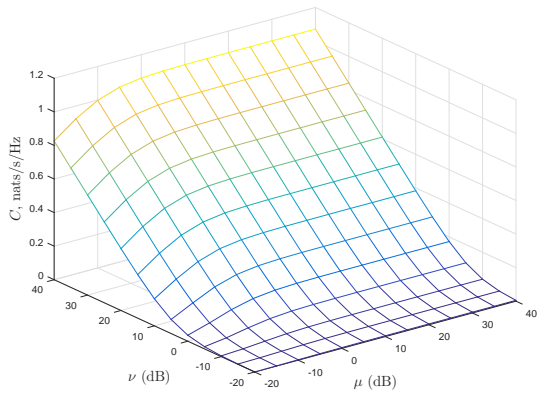}\label{fig_6and7combined1}

} \subfigure[]{
\includegraphics[scale=.39]{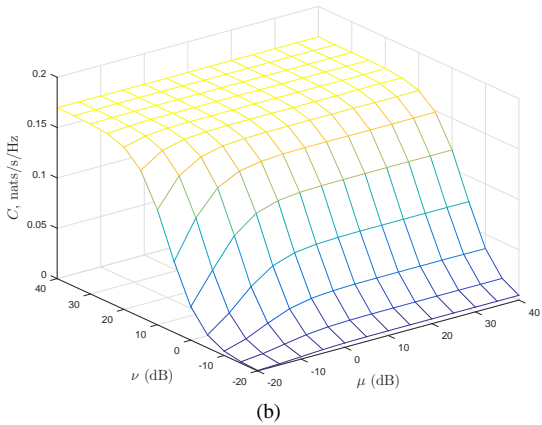}
\label{fig_6and7combined2}
}
\end{center}
\caption{Per-antenna ergodic capacity versus SNR levels   
$\mu$, $\nu$ ($\beta=5$, $\gamma=10$, $N=100$, $M=10$, $K=50$, (a)  $\delta_{t_1}=\delta_{t_2}=\delta_{r_1}=\delta_{r_2}=0$, (b) $\delta_{t_1}=\delta_{t_2}=\delta_{r_1}=\delta_{r_2}=0.08$)}
\vspace{-10 pt}
\end{figure}



\begin{figure}
\begin{center}
\subfigure[]{
\includegraphics[scale=.4]{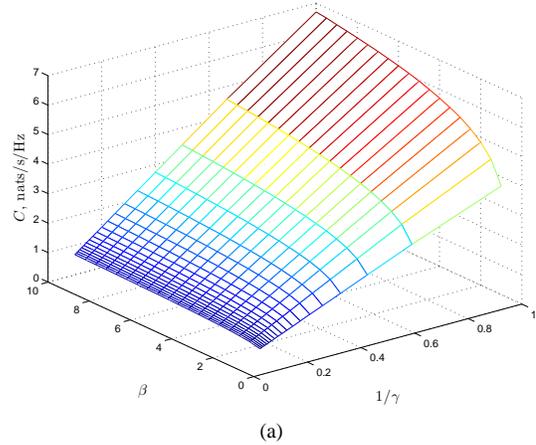}\label{fig_8and9combined1}

} \subfigure[]{
\includegraphics[scale=.4]{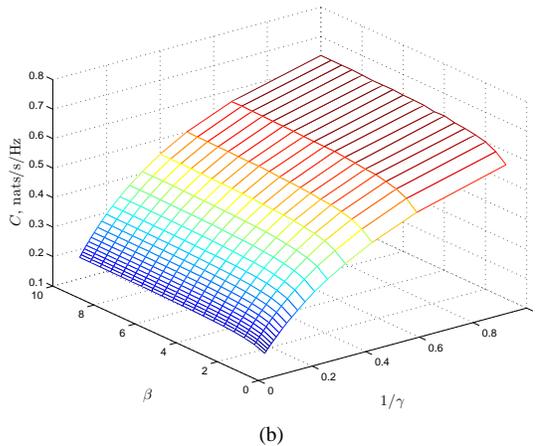}\label{fig_8and9combined2}

}
\end{center}
\caption{Per-antenna ergodic capacity versus channel dimensions $\beta$, $\gamma$ ($\mu=\nu=20$ dB, (a) $\delta_{t_1}=\delta_{t_2}=\delta_{r_1}=\delta_{r_2}=0$, (b) $\delta_{t_1}=\delta_{t_2}=\delta_{r_1}=\delta_{r_2}=0.08$)}
\vspace{-10 pt}
\end{figure}


\section{Conclusions}
\label{sec:6}
This paper  presented a thorough investigation of the impact of residual additive hardware impairment on the ergodic capacity of DH AF MIMO  relay channels by considering a large system  analysis. Specifically, we derived the ergodic capacity by employing the theory of FP. Moreover, we proceeded with the study of the effects of hardware impairments, when the number of antennas becomes large as massive MIMO architecture demands. Notably, we demonstrated the quantification of the degradation due to the additive RF impairments on the ergodic capacity by varying the  system parameters such as the SNR, the number of antennas, and the quality of the RF equipment. Nevertheless, the validation of the analytical results was shown  by means of simulations. In particular, simulations depicted that the asymptotic results can be applicable even for contemporary system dimensions.

\begin{appendices}
\section{Useful Lemmas}
Herein, given the eigenvalue probability distribution function $f_\bX(x)$ of a matrix  $\bX$, we provide useful definitions and lemmas, found in~\cite{Tulino04}, that are considered during our analysis. In the following definitions, $\delta$ is a nonnegative real number.
\begin{definition}[Shannon transform{\cite[Definition~2.12]{Tulino04}}]
\label{def: Shannon transform}
The Shannon transform of a positive semidefinite matrix \(\bX\) is defined as
\begin{align}
\mathcal{V}_{\bX}\left(\delta\right)=\int_0^\infty\ln\left({1+\delta x}\right)f_\mathbf{X}(x)\mathrm{d}x.
\end{align}
\end{definition}
\begin{definition}[$\eta$-transform{\cite[Definition~2.11]{Tulino04}}]
\label{def: n transform}
The $\eta$-transform of a positive semidefinite matrix \(\bX\) is defined as
\begin{align}
\eta_{\bX}\left(\delta\right)=\int_0^\infty\frac{1}{1+\delta x}f_\mathbf{X}(x)\mathrm{d}x.
\label{eq: eta def}
\end{align}
\end{definition}
\begin{definition}
\label{def: Sigma transform}[$\mathrm{S}$-transform{\cite[Definition~2.15]{Tulino04}}]
The $\mathrm{S}$-transform of a positive semidefinite matrix \(\bX\) is defined as\begin{align}
\Sigma_{\mathbf{X}}(x)=-\frac{x+1}{x}\eta^{-1}_{\mathbf{X}}(x+1).
\end{align}
\end{definition}
\begin{definition}[The Mar\v cenko-Pastur law density function {\cite{Bai1}}]\label{lemma:marcenko}
\label{df: MP law}
Given a Gaussian $K\times M$ channel matrix $\bH\sim\mathcal{C}\mathcal{N}\left( \mathbf{0},\Id \right)$, the a.e.p.d.f. of \(\frac{1}{K}\bH\bH^H\)
converges almost surely (a.s.) to the non-random limiting eigenvalue distribution
of the Mar\v cenko-Pastur law  given by
\begin{align}
        f^\infty_{\frac{1}{K}\bH\bH^H}(x)=\left(1-\beta\right)^+\left(x\right)+\frac{\sqrt{\left(x-a\right)^+\left(b-x\right)^+}}{2\pi x},\label{eq: MP aepdf}
\end{align}
where $a=(1-\sqrt{\beta})^2\mathrm{,}~b=(1+\sqrt{\beta})^2$, $\beta=\frac{K}{M}$, and $\delta\left( x \right)$ is Dirac's delta function.
\end{definition}

\begin{lemma}[{\cite[Eqs. 2.87, 2.88]{Tulino04}}]
\label{sigma transform N}
The $\mathrm{S}$-transform of the matrix  \(\frac{1}{K}\bH^{\H}\bH\) is expressed as
\begin{align}
\Sigma_{\frac{1}{K}\bH^H\bH}\left(x,\beta \right)&=\frac{1}{1+\beta x}\label{eq: sigma transform N1},
\end{align}
while the $\mathrm{S}$-transform of the matrix \(\frac{1}{K}\bH\bH^H\)  is obtained as
\begin{align}
\Sigma_{\frac{1}{K}\bH\bH^H}\left(x,\beta \right)&=
\frac{1}{\beta+x}.
\label{eq: sigma transform N},
\end{align}
\end{lemma}

\begin{lemma}[{\cite[Eq.~2.48]{Tulino04}}]
\label{pro: Sigma n}
The Stieltjes-transform of a positive semidefinite matrix \(\bX\) can be derived by its $\eta$-transform according to
\begin{align}
\mathcal{S}_{\mathbf{X}}(x)=-\frac{\eta_{\mathbf{X}}(-1/x)}{x}.
\end{align}
\end{lemma}
\begin{lemma}[{\cite[Eq.~2.45]{Tulino04}}]
\label{pro: Stieltjes to aepdf}
The a.e.p.d.f. of $\mathbf{X}$ is obtained by  the imaginary part of the Stieltjes transform \(\mathcal{S}\) for real arguments as
\begin{align}
f^{\infty}_\mathbf{X}(x)=\lim_{y\rightarrow0^+}\frac{1}{\pi}\mathfrak{I}\left\{ \mathcal{S}_\mathbf{X}(x+\mathrm{j}y) \right\}.
\label{eq: limiting eigenvalue pdf}
\end{align}
\end{lemma}
\section{Proof of Theorem~\ref{Capacity}}\label{ProofOfCapacity}

The asymptotic limits of the capacity terms~\eqref{eq:first form2} and~\eqref{eq:second form2}, when the channel dimensions tend to infinity, while keeping their  ratios  $\beta$ and $\gamma$ fixed,  are expressed by means of principles of FP theory in terms of a generic expression as 
\begin{align}
\mathrm{C}_i&\!=\!\frac{1}{N}\!\lim_{K,M,N\rightarrow\infty}\!\!\!\!\EE\!\left[\ln\det\!\left(\Id_{M}\!+\!f_{3}\bH_2^H\bH_2\!\left(\Id_{M}\!+\!\al\bH_\mathrm{1}\bH_\mathrm{1}^H\right)\!\right)\!\right]\nn\\
&\!=\!\frac{M}{N}\!\lim_{K,M,N\rightarrow\infty}\!\!\!\!\EE\!\left[\frac{1}{M}\!\sum_{i=1}^M \ln\!\left(\!1\!+\!f_{3}M\lambda_i\left(\frac{1}{M}\mathbf{K}_{\al}\!\right)\!\right)\!\right]\nn\\
&\!\rightarrow\!\frac{1}{\gamma}\!\int_0^\infty \ln\!\left(1\!+\!f_{3}M x\right)f^\infty_{\mathbf{K}_{\al}/M}\!\left(x\right)\mathrm{d}x,
\label{eq: capacity aepdf 1}
\end{align}
where $\mathrm{C}_i$ corresponds to $\mathrm{C}_1$ or $\mathrm{C}_2$ depending on the value of $i$, i.e., if $\al=f_1$ or if $\al=f_{2}/f_{3}$, respectively. In addition,   $\lambda_i\left(\mathbf{X}\right)$ is the $i$th ordered eigenvalue of matrix $\mathbf{X}$, and $f^\infty_{\mathbf{X}}$ denotes the asymptotic eigenvalue probability density function (a.e.p.d.f.) of $\mathbf{X}$. Moreover, for the sake of simplification of our analysis, we have made use of the following variable definitions similar to~\cite{Chatzinotas2013}
\begin{align}
 \tilde{\bM}_{\al}&=\Id_{M}+\al\bH_{1}\bH_{1}^{H}\\
 \tilde{\bN}_{1}&=\bH_{1}\bH_{1}^{\H}\\
   \tilde{\bN}_{2}&=\bH_{2}^{\H}\bH_{2}\\
{\bK}_{\al}&=\bH_{2}^{\H}\bH_{2}\left( \Id_{M}+\al\bH_{1}\bH_{1}^{H} \right)=\tilde{\bN}_{2}\tilde{\bM}_{\al}\label{ka}.
\end{align}

The a.e.p.d.f. of $\bK_{\al}/M$  can be obtained by means of Lemma~\ref{pro: Stieltjes to aepdf}, which demands its Stieltjes transform. In the following, we describe the steps leading to the derivation of the desired Stieltjes transform of $\bK_{\al}/M$. 
First, we employ Lemma~\ref{pro: Sigma n}, and take the inverse of $\eta$-transform of $\bK_{\al}/M$ as
\begin{align}
 x \eta^{-1}_{\bK_{\al}/M}\left(-x \mathcal{S}_{\bK_{\al}/M}\left( x \right) \right)+1=0.\label{lemma2a}
\end{align}
Thus, we now focus on the derivation of $ \eta^{-1}_{\bK_{\al}/M}\left(x  \right)$. 
\begin{proposition}
\label{Theorem: ieta transform K}
The inverse $\eta$-transform of $\mathbf{K}_{\al}/M$ is given by
\begin{align}
\eta^{-1}_{\mathbf{K}_{\al/M}}(x)=\Sigma_{\mathbf{\tilde N_{2}}/M}(x-1)\eta^{-1}_{\tilde{\bM}_{\al}/M}(x).
\label{eq: inverse eta transform K}
\end{align}
\begin{proof}
The inverse of the $\eta$-transform of $\mathbf{K}_{\al}/M$ is given by means of the free convolution 
\begin{align}
\!\!\Sigma_{\mathbf{K}_{\al}/M}(x)&\!=\!\Sigma_{\mathbf{\tilde N_{2}}/M}(x)\Sigma_{\tilde{\mathbf{M}}_{\al}/M}(x)\label{free convolution}\!\Longleftrightarrow\!\\
\left(\!-\frac{x\!+\!1}{x}\!\right)\!\eta^{-1}_{\mathbf{K}_{\al}/M}(x\!+\!1)&\!=\!\Sigma_{\mathbf{\tilde N_{2}}/M}(x)\!\left(\!-\frac{x\!+\!1}{x}\!\right)\!\eta^{-1}_{\tilde{\mathbf{M}}_{\al}/M}(x\!+\!1),\nonumber
\end{align}
where we have taken into advantage the asymptotic freeness between the deterministic matrix with bounded eigenvalues  $\mathbf{\tilde N_{2}}/M$ and the unitarily invariant matrix $\tilde{\bM}_{\al}/M $. Note that in~\eqref{free convolution}, we have applied Definition \ref{def: Sigma transform}. Appropriate change of variables, i.e., $y=x+1$ provides eq. \eqref{eq: inverse eta transform K}.
\end{proof}
\end{proposition}
\begin{longequation*}[tp]
\begin{align}        f^\infty_{\tilde{\bM}_{\al}/M}(x,\beta,\bar \al){\rightarrow}  {\frac {\sqrt { \left( x-1-\bar \al+2\bar  \al\sqrt {\beta}-\bar \al \beta \right) 
 \left( \bar \al+2\bar  \al\sqrt {\beta}+\bar \al \beta-x+1 \right) }}{2\bar \al\pi   \left( x-1
 \right) }}.
 \label{eq: I plus MP aepdf}
\end{align}
\hrule
\end{longequation*}
In order to obtain  $\eta^{-1}_{\tilde{\bM}_{\al}/M}\left(x  \right)$, we first need its a.e.p.d.f., given by the next proposition.
\begin{proposition}
 \label{lem: M aepdf}
The a.e.p.d.f. of $\tilde{\bM}_{\al}/M$ converges almost surely to~\eqref{eq: I plus MP aepdf} with  $\bar \al=M \al$.
\end{proposition}

\begin{proof}By denoting $z$ and $x$  the eigenvalues of $  \tilde{\bM}_{\al}/M$ and $\frac{1}{M}\tilde{\bN}_{1}$, respectively, the a.e.p.d.f. of $  \tilde{\bM}_{\al}/M$ can be obtained after making the transformation  $z(x)=(1+M \al x)$ as
\begin{align}
f^\infty_{\tilde{\bM}_{\al}/M}(z)&=\left| \frac{1}{z'(z^{-1}(x))} \right|\! \cdot \! f^\infty_{\frac{1}{M}\tilde{\bN}_{1}}\left( z^{-1}(x) \right)\nn\\
&=\frac{1}{\bar \al}f^\infty_{\frac{1}{M}\tilde{\bN}_{1}}\left(\frac{z-1}{\bar \al}\right).
\label{eq: M aepdf}
\end{align}
\end{proof}
\begin{longequation*}[tp]
\begin{align}
\!\!\!\eta_\mathbf{\tilde{M}_{\al}/M}(\psi)&\!=\!\frac{\bar\al}{4i\pi}\!\oint_{\left|\zeta\right|=1}\!\frac{(\zeta^2-1)^2}{\zeta((1+\beta)\zeta\!+\!\sqrt{\beta}(\zeta^2+1))(\zeta(1+\psi(1+\bar\al+\bar\al\beta))\!+\!\sqrt{\beta}\psi\bar\al(\zeta^2+1))}d\zeta.\!
\label{eq: eta derivation}
\end{align}
\hrule
\end{longequation*}
Consequently, we are ready to obtain $\eta^{-1}_{\tilde{\mathbf{M}}_{\al}/M}(x)$ by using~\eqref{eq: eta def}.
\begin{proposition}
\label{Theorem: eta transform M}
The inverse $\eta$-transform of $\tilde{\mathbf{M}}_{\al}/M$ is given by \eqref{eq: inverse eta transform M}.
\end{proposition}
\begin{proof}
\begin{longequation*}[tp]
\begin{align}
\!\!\!\!\!\eta^{\!-1}_{\tilde{\mathbf{M}}_{\al}/M}(x)\!=\!{\frac {-x\bar \al\!-\!\beta\bar \al\!+\!\bar \al\!-\!1\!\!+\!\!\sqrt {\!{x}^{2}{\bar \al}^{2
}\!+\!2x\bar {\al}^{2}\beta\!-\!2x\bar {\al}^{2}\!-\!2x\bar \al\!+\!{\beta}^{2}{\bar \al
}^{2}\!-\!2\beta{\bar \al}^{2}\!+\!2\beta\bar \al\!+\!{\bar \al}^{2}\!+\!2\bar \al\!+\!1
}}{2x\bar \al}}.\!\!\!\!
\label{eq: inverse eta transform M}
\end{align}
\hrule
\end{longequation*}
Having obtained the a.e.p.d.f. of $\tilde{\bM}_{\al}$, use of Definition~\ref{def: n transform} allows to derive its $\eta$-transform  as
\begin{align}
\!\!\!\eta_\mathbf{\tilde{M}_{\al}/M}(\psi)&\!=\!\!\int_{0}^{{+\infty}}\frac{1}{1+\psi x}f^\infty_\mathbf{\tilde{M}_{\al}/M}(x)dx. \nonumber
\end{align}
If we make the necessary substitution, $\eta_\mathbf{\tilde{M}_{\al}/M}(\psi)$ is written as in~\eqref{eq: eta derivation}. Following a similar procedure as in~\cite{Bai2010}, we  perform certain  substitutions. Specifically, we set  $x=w\bar\al+1$, $dx=\bar\al dw$, followed by $w=1+\beta+2\sqrt{\beta}\cos \omega$, $dw=2\sqrt{\beta}(-\sin\omega)d\omega$, and finally $\zeta=e^{i\omega}$,  $d\zeta=i\zeta d\omega$. Hence, initially we  calculate the poles $\zeta_i$ and residues $\rho_i$ of Eq.~\eqref{eq: eta derivation}.
Then, we perform an appropriate Cauchy integration by including the residues  located within the unit disk. More concretely, we have
\begin{align}
\eta_\mathbf{\tilde{M}_{\al}/M}(\psi)=-\frac{\beta}{2}(\rho_0+\rho_2+\rho_4),\nonumber                                                                     \end{align}
which after inversion results to Eq.~\eqref{eq: inverse eta transform M}.
\end{proof}

As far as $\Sigma_{\mathbf{\tilde N_{2}}/M}(x)$ is concerned, it is given by~\eqref{eq: sigma transform N} as
\begin{align}
 \Sigma_{\mathbf{\tilde N_{2}}/M}(x)=\frac{1}{\gamma+ x}.\label{N2}
\end{align}

In the last step,  having calculated  $\eta^{-1}_{\mathbf{K}_{\al}}(x)$ from~\eqref{eq: inverse eta transform K} after substituting~\eqref{eq: inverse eta transform M} and~\eqref{N2}, we employ~\eqref{lemma2a}, and after tedious algebraic manipulations we obtain the following quartic polynomial
\begin{align}
\bar{\al}^{2}x^2 \mathcal{S}_{\bK_{\al}/M}^{4} \nn \\
+(2 \bar{\al}^2 (1-\gamma) x+\bar{\al}^2 x^2) {S}_{\bK_{\al}/M}^{3} \nn \\
+(\bar{\al}^2 (2-\beta-\gamma) x+\bar{\al}^2 (\gamma-1)^2-\bar{\al} x) {S}_{\bK_{\al}/M}^{2} \nn \\
+(\bar{\al}^2 (\beta (\gamma-1)-\gamma)+\bar{\al} (\gamma+\bar{\al}-x-1) {S}_{\bK_{\al}/M}-\bar{\al}.
\label{eq: polK}
\end{align}


\end{appendices}

\section*{Acknowledgement}
\vspace{-5 pt}
This work was supported by  a Marie Curie Intra European Fellowship within the 7th European Community Framework Programme for Research of the European Commission under grant agreements no. [330806], ``IAWICOM'', and partially by FNR, Luxembourg under the CORE projects ``SeMIGod'' and ``SATSENT''.

\bibliographystyle{IEEEtran}

\bibliography{mybib}

\begin{thebibliography}{10}
\providecommand{\url}[1]{#1}
\csname url@samestyle\endcsname
\providecommand{\newblock}{\relax}
\providecommand{\bibinfo}[2]{#2}
\providecommand{\BIBentrySTDinterwordspacing}{\spaceskip=0pt\relax}
\providecommand{\BIBentryALTinterwordstretchfactor}{4}
\providecommand{\BIBentryALTinterwordspacing}{\spaceskip=\fontdimen2\font plus
\BIBentryALTinterwordstretchfactor\fontdimen3\font minus
  \fontdimen4\font\relax}
\providecommand{\BIBforeignlanguage}[2]{{%
\expandafter\ifx\csname l@#1\endcsname\relax
\typeout{** WARNING: IEEEtran.bst: No hyphenation pattern has been}%
\typeout{** loaded for the language `#1'. Using the pattern for}%
\typeout{** the default language instead.}%
\else
\language=\csname l@#1\endcsname
\fi
#2}}
\providecommand{\BIBdecl}{\relax}
\BIBdecl

\bibitem{Telatar1999}
E.~Telatar, ``Capacity of multi-antenna {G}aussian channels,'' \emph{Europ.
  Trans. on Telecom.}, vol.~10, no.~6, pp. 585--595, 1999.

\bibitem{Foschini1998}
G.~J. Foschini and M.~J. Gans, ``On limits of wireless communications in a
  fading environment when using multiple antennas,'' \emph{Wireless Pers.
  Commun.}, vol.~6, no.~3, pp. 311--335, 1998.

\bibitem{Marzetta2010}
T.~Marzetta, ``Noncooperative cellular wireless with unlimited numbers of base
  station antennas,'' \emph{IEEE Trans. Wireless Commun.}, vol.~9, no.~11, pp.
  3590--3600, November 2010.

\bibitem{Couillet2011}
R.~Couillet and M.~Debbah, \emph{Random matrix methods for wireless
  communications}.\hskip 1em plus 0.5em minus 0.4em\relax Cambridge University
  Press, 2011.

\bibitem{Papazafeiropoulos2014}
A.~Papazafeiropoulos and T.~Ratnarajah, ``Uplink performance of massive {MIMO}
  subject to delayed {CSIT} and anticipated channel prediction,'' in \emph{IEEE
  International Conference on Acoustics, Speech and Signal Processing (ICASSP),
  2014}, May 2014, pp. 3162--3165.

\bibitem{Papazafeiropoulos2015a}
------, ``Deterministic equivalent performance analysis of time-varying massive
  {MIMO} systems,'' \emph{to be published in IEEE Trans. Wireless Commun.}, pp.
  1--15, 2015.

\bibitem{Papazafeiropoulos2015}
A.~Papazafeiropoulos, ``Impact of user mobility on optimal linear receivers in
  cellular networks,'' in \emph{Proc. IEEE Int. Conf. Commun.}, London, Jun.
  2015, pp. 3842--3849.

\bibitem{Qi2010}
J.~Qi and S.~A{\"\i}ssa, ``Analysis and compensation of {I/Q} imbalance in
  {MIMO} transmit-receive diversity systems,'' \emph{IEEE Trans. Commun.},
  vol.~58, no.~5, pp. 1546--1556, 2010.

\bibitem{Qi2012}
------, ``On the power amplifier nonlinearity in {MIMO} transmit beamforming
  systems,'' \emph{IEEE Trans. Commun.}, vol.~60, no.~3, pp. 876--887, 2012.

\bibitem{Schenk2008}
T.~Schenk, \emph{RF imperfections in high-rate wireless systems: impact and
  digital compensation}.\hskip 1em plus 0.5em minus 0.4em\relax Springer
  Science \& Business Media, 2008.

\bibitem{Studer2010}
C.~Studer, M.~Wenk, and A.~Burg, ``{MIMO} transmission with residual
  transmit-rf impairments,'' in \emph{ITG/IEEE Work. Smart Ant. (WSA)}.\hskip
  1em plus 0.5em minus 0.4em\relax IEEE, 2010, pp. 189--196.

\bibitem{Bjoernson2013}
E.~Bj{\"o}rnson, P.~Zetterberg, M.~Bengtsson, and B.~Ottersten, ``Capacity
  limits and multiplexing gains of {MIMO} channels with transceiver
  impairments,'' \emph{IEEE Commun. Lett.}, vol.~17, no.~1, pp. 91--94, 2013.

\bibitem{Zhang2014}
X.~Zhang, M.~Matthaiou, E.~Bj{\"o}rnson, M.~Coldrey, and M.~Debbah, ``On the
  {MIMO} capacity with residual transceiver hardware impairments,'' in \emph{in
  Proc. IEEE Int. Conf. Commun.}\hskip 1em plus 0.5em minus 0.4em\relax IEEE,
  2014, pp. 5299--5305.

\bibitem{Bjornson2013}
E.~Bj{\"o}rnson, M.~Matthaiou, and M.~Debbah, ``A new look at dual-hop
  relaying: Performance limits with hardware impairments,'' \emph{IEEE Trans.
  Commun.}, vol.~61, no.~11, pp. 4512--4525, 2013.

\bibitem{Yang2009}
Y.~Yang, H.~Hu, J.~Xu, and G.~Mao, ``Relay technologies for {WiMAX} and
  {LTE}-advanced mobile systems,'' \emph{IEEE Commun. Mag.}, vol.~47, no.~10,
  pp. 100--105, 2009.

\bibitem{DBLP:journals/twc/ChatzinotasIH09}
S.~Chatzinotas, M.~A. Imran, and R.~Hoshyar, ``On the multicell processing
  capacity of the cellular {MIMO} uplink channel in correlated rayleigh fading
  environment,'' \emph{{IEEE} Transactions on Wireless Communications}, vol.~8,
  no.~7, pp. 3704--3715, 2009.

\bibitem{Holma2011}
H.~Holma and A.~Toskala, \emph{LTE for UMTS: Evolution to LTE-Advanced}, Wiley,
  Ed., 2011.

\bibitem{Jin2010}
S.~Jin, M.~R. McKay, C.~Zhong, and K.~K. Wong, ``Ergodic capacity analysis of
  amplify-and-forward{ MIMO} dual-hop systems,'' \emph{IEEE Trans. on Inform.
  Theory}, vol.~56, no.~5, pp. 2204--2224, 2010.

\bibitem{Tulino04}
A.~M. Tulino and S.~Verd{\'u}, \emph{Random matrix theory and wireless
  communications}.\hskip 1em plus 0.5em minus 0.4em\relax Now Publishers Inc.,
  2004, vol.~1, no.~1.

\bibitem{Bai1}
J.~W. Silverstein and Z.~Bai, ``On the empirical distribution of eigenvalues of
  a class of large dimensional random matrices,'' \emph{Journal of Multivariate
  analysis}, vol.~54, no.~2, pp. 175--192, 1995.

\bibitem{Chatzinotas2013}
S.~Chatzinotas, ``{MMSE} filtering performance of dual-hop amplify-and-forward
  multiple-access channels,'' \emph{IEEE Wireless Commun. Lett.}, vol.~2,
  no.~1, pp. 122--125, 2013.

\bibitem{Bai2010}
Z.~Bai and J.~W. Silverstein, \emph{Spectral analysis of large dimensional
  random matrices}.\hskip 1em plus 0.5em minus 0.4em\relax Springer, 2010,
  vol.~20.

\end{thebibliography}
\end{document}